\documentclass[pdftex]{llncs}
\usepackage[T1]{fontenc}
\usepackage{ae,aecompl}
\usepackage[english]{babel}
\usepackage{amssymb}

\usepackage{color,soul}
\usepackage{subfigure}
\usepackage{graphicx}
\usepackage{multirow}

\title{Using cascading Bloom filters to improve the memory usage for de Brujin graphs}

\author{
Kamil Salikhov\inst{1}
\and
Gustavo Sacomoto\inst{2,3}
\and
Gregory Kucherov\inst{4,5}
}
\institute{
Lomonosov Moscow State University, Moscow, Russia, \email{salikhov.kamil@gmail.com}
\and
INRIA Grenoble Rh\^one-Alpes, France, \email{gustavo.sacomoto@inria.fr}
\and
Laboratoire Biom\'etrie et Biologie Evolutive, Universit\'e Lyon 1, Lyon, France
\and
Department of Computer Science, Ben-Gurion University of the Negev, Be'er Sheva, Israel
\and
Laboratoire d'Informatique Gaspard Monge, Universit\'e Paris-Est \& CNRS, Marne-la-Vall\'ee, Paris, France, \email{Gregory.Kucherov@univ-mlv.fr}
}

\begin{document}

\maketitle

\begin{abstract}
De Brujin graphs are widely used in bioinformatics for processing
next-generation sequencing data. Due to a very large size of NGS
datasets, it is essential to represent de Bruijn graphs compactly, and
several approaches to this problem have been proposed recently.  In
this work, we show how to reduce the memory required by the algorithm
of \cite{DBLP:conf/wabi/ChikhiR12} that represents de Brujin graphs
using Bloom filters. Our method requires 30\% to 40\% less memory with
respect to the method of \cite{DBLP:conf/wabi/ChikhiR12}, with
insignificant impact to construction time. At the same time, our
experiments showed a better query time compared to
\cite{DBLP:conf/wabi/ChikhiR12}. This is, to our
knowledge, the best \emph{practical} representation for de Bruijn
graphs.
\end{abstract}

\section{Introduction}
Modern Next-Generation Sequencing (NGS) technologies generate huge
volumes of short nucleotide sequences (\emph{reads}) drawn from the
DNA sample under study. The length of a read varies from 35 to about
400 base pairs (letters) and the number of reads may be hundreds of
millions, thus the total volume of data may reach tens or even
hundreds of Gb.

Many computational tools dealing with NGS data, especially those
devoted to {\em genome assembly}, are based on the concept of \emph{de
  Bruijn graph}, see e.g. \cite{pmid20211242}. The nodes of the de
Bruijn graph are all distinct {\em $k$-mers} occurring in the reads,
and two {\em $k$-mers} are linked by an edge if there is a
suffix-prefix overlap of size $k-1$.\footnote{Note that this actually
  a {\em subgraph} of the de Bruijn graph under its classical
  combinatorial definition. However, we still call it de Bruijn graph
  to follow the terminology common to the bioinformatics literature.}
The value of $k$ is an open parameter, that in practice is chosen
between $20$ and $64$. The idea of using de Bruijn graph for genome
assembly goes back to the ``pre-NGS era'' \cite{pmid11504945}.  Note,
however, that {\em de novo} genome assembly is not the only
application of those graphs when dealing with NGS data. There are
several others, including: {\em de novo} transcriptome
assembly~\cite{trinity} and {\em de novo} alternative splicing
calling~\cite{kissplice} from transcriptomic NGS data (RNA-seq);
metagenome assembly~\cite{meta-idba} from metagenomic NGS data; and
genomic variant detection~\cite{cortex} from genomic NGS data using a
reference genome.

Due to the very large size of NGS datasets, it is essential to
represent de Bruijn graphs as compactly as possible.  This has been a
very active line of research. Recently, several papers have been
published that propose different approaches to compressing de Bruijn
graphs
\cite{Conway15022011,YeAtAl-BMCBioinfo12,DBLP:conf/wabi/ChikhiR12,DBLP:conf/wabi/BoweOSS12,pmid22847406}.

Conway and Bromage~\cite{Conway15022011} proposed a method based on
classical succinct data structures, i.e. bitmaps with efficient
rank/select operations.  On the same direction, Bowe \emph{et
  al.}~\cite{DBLP:conf/wabi/BoweOSS12} proposed a very interesting
succinct representation that, assuming only one string (read) is
present, uses only $4m$ bits, where $m$ is the number of edges in the
graph. The more realistic case, where there are $M$ reads, can be
easily reduced to the one string case by concatenating all $M$ reads
using a special separator character. However, in this case the size of
the structure is $4m + O(M \log m)$ bits
(\cite{DBLP:conf/wabi/BoweOSS12}, Theorem 1). Since the multiplicative
constant of the second term is hidden by the asymptotic notation, it
is hard to know precisely what would be the size of this structure in
practice.

Ye at al.~\cite{YeAtAl-BMCBioinfo12} proposed a different method based
on a sparse representation of de Bruijn graphs, where only a subset of
$k$-mers present in the dataset are stored. Pell et
al.~\cite{pmid22847406} proposed a method to represent it
approximately, the so called \emph{probabilistic de Bruijn graph}. In
their representation a node have a small probability to be a false
positive, i.e. the $k$-mer is not present in the dataset. Finally,
Chikhi and Rizk~\cite{DBLP:conf/wabi/ChikhiR12} improved Pell's scheme
in order to obtain an exact representation of the de Bruijn
graph. This was, to our knowledge, the best \emph{practical}
representation of an exact de Bruijn graph.

In this work, we focus on the method proposed in
\cite{DBLP:conf/wabi/ChikhiR12} which is based on Bloom filters.  They
were first used in \cite{pmid22847406} to provide a very
space-efficient representation of a subset of a given set (in our
case, a subset of $k$-mers), at the price of allowing {\em one-sided
  errors}, namely {\em false positives}. The method of
\cite{DBLP:conf/wabi/ChikhiR12} is based on the following idea: if all
queried nodes ($k$-mers) are only those which are reachable from some
node known to belong to the graph, then only a fraction of all false
positives can actually occur. Storing these false positives explicitly
leads to an exact (false positive free) and space-efficient
representation of the de Bruijn graph.

Our contribution is an improvement of this scheme by changing the
representation of the set of false positives. We achieve this by
iteratively applying a Bloom filter to represent the set of false
positives, then the set of ``false false positives'' etc. We show
analytically that this cascade of Bloom filters allows for a
considerable further economy of memory, improving the method of
\cite{DBLP:conf/wabi/ChikhiR12}. Depending on the value of $k$, our
method requires 30\% to 40\% less memory with respect to the method of
\cite{DBLP:conf/wabi/ChikhiR12}. Moreover, with our method, the memory
grows very little with the growth of $k$. Finally, we implemented our
method and tested it against \cite{DBLP:conf/wabi/ChikhiR12} on real
datasets. The tests confirm the theoretical predictions for the size
of structure and show a 20\% to 30\% \emph{improvement} in query
times.

\section{Cascading Bloom filter}

Let $T_0$ be the set of $k$-mers (nodes) of the de Brujin graph that we want
to store. The method of \cite{DBLP:conf/wabi/ChikhiR12} stores $T_0$
via a bitmap $B_1$ using a Bloom filter, together with the set $T_1$
of {\em critical false positives}. $T_1$ consists of those $k$-mers
which overlap $k$-mers from $T_0$ by $k-1$ letters (i.e. are
``potential neighbors'' of nodes of $T_0$) but which are stored in
$B_1$ "by mistake", i.e. which belong to $B_1$ but are not in
$T_0$.\footnote{By a slight abuse of language, we say that ``an
  element belongs to $B_j$'' if it is accepted by the corresponding
  Bloom filter.} $B_1$ and $T_1$ are sufficient to represent the graph
provided that the only queried $k$-mers are those which are potential neighbors of $k$-mers of $T_0$.

The idea we introduce in this note is to use this structure recursively and represent the set $T_1$ by a new bitmap $B_2$ and a new set $T_2$, then represent $T_2$ by $B_3$ and $T_3$, and so on.
More formally, starting from $B_1$ and $T_1$ defined as above, we define a series of bitmaps $B_1, B_2, \ldots$ and a series of sets $T_1,T_2,\ldots$
as follows.
$B_2$ will store the set $T_1$ using a Bloom filter, and the set $T_2$ will contain ``true nodes'' from $T_0$ that are stored in $B_2$ ``by mistake'' (we call them {\bf false$^{2}$} positives). $B_3$ and $T_3$, and, generally, $B_i$ and $T_i$ are defined similarly: $B_i$ stores $k$-mers of $T_{i-1}$ using a Bloom filter, and $T_i$ contains $k$-mers stored in $B_i$ "by mistake", i.e. those $k$-mers that do not belong to $T_{i-1}$ but belong to $T_{i-2}$ (we call them {\bf false}$^{i}$ positives). Observe that $T_0\cap T_1=\emptyset$, $T_0 \supseteq T_2 \supseteq T_4 \ldots$ and $T_1 \supseteq T_3 \supseteq T_5 \ldots$.

The following lemma shows that the construction is correct, that is it allows one to verify whether or not a given $k$-mer is belongs to the set $T_0$.

\begin{lemma}
\label{mainlemma}
Given a $k$-mer (node) $K$, consider the smallest $i$ such that $K \not\in B_{i+1}$ (if $K \not\in B_1$, we define $i=0$).
Then, if $i$ is odd, then $K\in T_0$, and if $i$ is even (including $0$), then $K\not\in T_0$.
\end{lemma}
\begin{proof}
Observe that $K\not\in B_{j+1}$ implies $K\not\in T_j$ by the basic property of Bloom filters. We first check the Lemma for $i=0,1$.

For $i=0$,  we have $K \not\in B_1$, and then $K \not\in T_0$.
%

For $i=1$, we have $K\in B_1$ but $K\not\in B_2$. The latter implies that $K\not\in T_1$, and then $K$ must be a false$^2$ positive, that is $K\in T_0$.
Note that here we use the fact that the only queried $k$-mers $K$ are either nodes of $T_0$ or their neighbors in the graph (see \cite{DBLP:conf/wabi/ChikhiR12}), and therefore if $K\in B_1$ and $K\not\in T_0$ then $K\in T_1$.


For the general case $i \geq 2$, we show by induction that $K\in T_{i-1}$. Indeed, $K\in B_1\cap \ldots\cap B_i$ implies $K\in T_{i-1}\cup T_i$ (which, again, is easily seen by induction),
and $K\not\in B_{i+1}$ implies $K\not\in T_i$.

Since $T_{i-1}\subseteq T_0$ for odd $i$, and $T_{i-1}\subseteq T_1$ for even $i$ (for $T_0\cap T_1=\emptyset$), the lemma follows.
\end{proof}

Naturally, the lemma provides an algorithm to check if a given $k$-mer $K$ belongs to the graph: it suffices to check successively if it belongs to $B_1,B_2,\ldots$ until we encounter the first $B_{i+1}$ which does not contain $K$. Then the answer will simply depend on whether $i$ is even or odd.

In our reasoning so far, we assumed an infinite number of bitmaps $B_i$.
Of course, in practice we cannot store infinitely many (and even simply many) bitmaps. Therefore we ``truncate'' the construction at some step $t$ and store a finite set of bitmaps $B_1, B_2, \ldots,B_t$ together with an explicit representation of $T_t$. The procedure of Lemma~\ref{mainlemma} is extended in the obvious way: if for all $1\leq i \leq t$, $K\in B_i$, then the answer is determined by directly checking $K\in T_t$.

\section{Memory and time usage} \label{sec:memory_time}

First, we estimate the memory needed by our data structure, under the assumption of infinite number of bitmaps. Let $N$ be the number of ``true positives'', i.e. nodes of $T_0$. From properties of Bloom filters, if $T_0$ has to be stored via a bitmap $B_1$ of size $rN$, the false positive rate can be estimated as $c^r$, where $c=0.6185$ (see \cite{Kirsch:2008:LHS:1400123.1400125}). From this, the expected number of critical false positive nodes (set $T_1$) have been estimated in \cite{DBLP:conf/wabi/ChikhiR12} to be $8Nc^r$, as every node has eight extensions, i.e. potential neighbors in the graph. We slightly refine this estimation to $6Nc^r$ by noticing that for almost all nodes, two out of these eight extensions belong to $T_0$ and only six are potential false positives.
Furthermore, to store these $6Nc^r$ critical false positive nodes, we use a bitmap $B_2$ of size $6rNc^r$. Bitmap $B_3$ is used for storing nodes of $T_0$ which are stored in $B_2$ ``by mistake'' (set $T_2$).
We estimate the number of these nodes as the fraction $c^r$ (false positive rate of filter $B_2$) of $N$ (size of $T_0$), that is $Nc^r$.
Similarly, the number of nodes we need to put to $B_4$  is $6Nc^r$
multiplied by $c^r$, i.e. $6Nc^{2r}$. Keeping counting in this way, we
obtain that the memory needed for the whole structure is $rN + 6rNc^r
+ rNc^r+ 6rNc^{2r} + rNc^{2r} + ...$ bits. By dividing this expression
by $N$ to obtain the number of bits per one $k$-mer, we get 
\begin{equation}r + 6rc^r + rc^r + 6rc^{2r} + ... = (r+6rc^r )(1 + c^r + c^{2r}+...)= (1+6c^r) \frac{r}{1-c^r}.\label{formula}\end{equation} A simple calculation shows that the minimum of this expression is achieved when $r = 5.464$, and then the minimum memory used per $k$-mer is $8.45$ bits.

As mentioned earlier, in practice we store only a finite number of bitmaps $B_1,\ldots,B_t$ together with an explicit representation (such as array or hash table) of $T_t$.  
In this case, the memory taken by the bitmaps is a truncated sum $rN + 6rNc^r + rNc^r+ ..$, and a data structure storing $T_t$ takes either $2k \cdot Nc^{\lceil\frac{t}{2}\rceil r}$ or $2k \cdot 6Nc^{\lceil\frac{t}{2}\rceil r}$ bits, depending on whether $t$ is even or odd. The latter follows from the observations that we need to store $Nc^{\lceil\frac{t}{2}\rceil r}$ (or $6rNc^{\lceil\frac{t}{2}\rceil r}$) $k$-mers, each taking $2k$ bits of memory. Consequently, we have to adjust the optimal value of $r$ minimizing the total space, and re-estimate the resulting space spent on one $k$-mer.
In this case, the memory taken by the bitmaps is a truncated sum $rN + 8rNc^r + rNc^r+ ..$, and a data structure storing $T_t$ takes either $2k \cdot Nc^{\lceil\frac{t}{2}\rceil r}$ or $2k \cdot 8Nc^{\lceil\frac{t}{2}\rceil r}$ bits, depending on whether $t$ is even or odd. The latter follows from the observations that we need to store $Nc^{\lceil\frac{t}{2}\rceil r}$ (or $8rNc^{\lceil\frac{t}{2}\rceil r}$) $k$-mers, and every $k$-mer takes $2k$ bits of memory. Consequently, we have to adjust the optimal value of $r$ minimizing the total space, and re-estimate the resulting space spent on one $k$-mer.


Table~\ref{table1} shows estimations for optimal values of $r$ and the corresponding space per $k$-mer for $t=4$ and several values of $k$. The data demonstrates that even such a small value of $t$ leads to considerable memory savings. It appears that the space per $k$-mer is very close to the ``optimal'' space ($8.45$ bits) obtained for the infinite number of filters. Table~\ref{table1} reveals another advantage of our improvement: the number of bits per stored $k$-mer remains almost constant for different values of $k$.

\begin{table}
\begin{center}
\begin{tabular}{|c|c|c|c|}
\hline
$k$ &optimal $r$ & bits per $k$-mer & bits per $k$-mer\\
 & for $t=4$ & for $t=4$ & for $t=1$ (\cite{DBLP:conf/wabi/ChikhiR12}) \\\hline\hline
16 & 5.777 & 8.556& 12.078 \\
\hline
32 & 6.049 & 8.664& 13.518 \\
\hline
64 & 6.399 & 8.824& 14.958 \\
\hline
128 & 6.819 & 9.045& 16.398\\
\hline
\end{tabular}
\end{center}
\caption{1st column: $k$-mer size; 2nd column: optimal value of $r$ for Bloom filters (bitmap size per number of stored elements) for $t=4$; 3rd column: the resulting space per $k$-mer; 4th column: space per $k$-mer for the method of \cite{DBLP:conf/wabi/ChikhiR12} ($t=1$)
}\label{table1}
\end{table}

The third column of Table~\ref{table1} shows the memory usage of the original method of \cite{DBLP:conf/wabi/ChikhiR12}, obtained using the estimation 
$(1.44 \log_2 (\frac{16k}{2.08}) + 2.08)$ from \cite{DBLP:conf/wabi/ChikhiR12}. 
Note that according to the method of \cite{DBLP:conf/wabi/ChikhiR12},
doubling the value of $k$ results in a memory increment by $1.44$
bits, whereas in our method the increment is of $0.11$ to $0.22$
bits. 
%

\bigskip
Let us now estimate preprocessing and query times for our scheme. If the value of $t$ is small (such as $t=4$, as in  Table~\ref{table1}), the preprocessing time grows insignificantly in comparison to the original method of \cite{DBLP:conf/wabi/ChikhiR12}. To construct each $B_i$, we need to store $T_{i-2}$ (possibly on disk, if we want to save on the internal memory used by the algorithm) in order to compute those $k$-mers which are stored in $B_{i-1}$ ``by mistake''.
The preprocessing time increases little in comparison to the original method of \cite{DBLP:conf/wabi/ChikhiR12}, as the size of $B_i$ decreases exponentially and then the time spent to construct the whole structure is linear on the size of $T_0$.


Clearly, applying $t$ Bloom filters instead of a single one introduces
a multiplicative factor of $t$. On the other hand, set $T_t$
is generally much smaller than $T_0$, due to the above-mentioned
exponential decrease. Depending on the data structure for storing
$T_t$, the time saving in querying $T_t$ vs. $T_0$ may even dominate
the time loss in querying multiple Bloom filters. Our experimental
results (Section~\ref{sec:imple} below) confirm that this situation
does indeed occur in practice. Note that even in the case when
querying $T_t$ weakly depends on its size (e.g. when $T_t$ is
implemented by a hash table), the query time will not increase much,
due to a small value of $t$, as discussed earlier. 




\subsection{Using different values of $r$ for different filters} \label{subsec:difr}
In the previous section, we assumed that each of our Bloom filters
uses the same value of $r$, the ratio of bitmap size to the number of
stored $k$-mers. However, formula (\ref{formula}) 
for the number of bits per $k$-mer 
shows a difference for odd and even filter indices. This suggests
that using different parameters $r$ for different filters, rather than
the same for all filters, may reduce the space even further. 
If $r_i$ denotes the corresponding ratio for filter $B_i$, then
(\ref{formula}) should be rewritten to 
\begin{equation}
r_1 + 6r_2c^{r_1} + r_3c^{r_2} + 6r_4c^{r_1+r_3} + ...,
\end{equation}
and the minimum value of this expression becomes $7.913$.

In the same way, we can use different values of $r_i$ in the truncated
case. This leads to a small $2\%$ to $4\%$ improvement in
comparison with case of unique value of $r$. 
Table~\ref{table1} shows results for the case $t=4$ for different values of $k$.

\begin{table}[h]
\begin{center}
\begin{tabular}{|c|c|c|}
\hline
$k$ & bits per $k$-mer & bits per $k$-mer \\
       &  different values of $r$ & single value of $r$ \\\hline\hline
16 & 8.336 &  8.556 \\
\hline
32 & 8.404 & 8.664\\
\hline
64 & 8.512 & 8.824\\
\hline
128 & 8.669 & 9.045\\
\hline
\end{tabular}
\end{center}
\label{table3}
\caption{Estimated memory occupation for the case of different values
  of $r$ vs. single value of $r$, for 4 Bloom filters ($t=4$). For the
  case of single $r$, its value is shown in Table~\ref{table1}. }
\end{table}

\section{Experimental results}
\subsection{Implementation and experimental setup} \label{sec:imple}
We implemented our method using the {\sc Minia} software 
\cite{DBLP:conf/wabi/ChikhiR12} and ran comparative tests for
$2$ and $4$ Bloom filters ($t = 2,4$). 
Note that since the only modified part was the construction step and
the $k$-mer membership queries, this allows us to precisely evaluate
our method against the one of \cite{DBLP:conf/wabi/ChikhiR12}. 

The first step of the implementation is to retrieve the list of $k$-mers
that appear more than $d$ times using
DSK~\cite{dsk} -- a constant memory streaming algorithm to count
$k$-mers.
Each $k$-mer appearing
more than $d$ times (set $T_0$) is inserted into $B_1$. Next, all possible extensions of each
$k$-mer in $T_0$ are queried against $B_1$, and those which return
true are written on the disk. 
Then, this set is traversed and only the $k$-mers absent from $T_0$ are
kept. This results in the set $T_1$ of critical false positives, which is also kept on
disk. Up to this point, the procedure is identical to that of
\cite{DBLP:conf/wabi/ChikhiR12}. 

Next, we insert all $k$-mers from $T_1$ into $B_2$ and to
obtain $T_2$, we check for each $k$-mer in $T_0$ if a query to $B_2$
returns true. This results in the set $T_2$. Thus, 
at this point we have $B_1$, $B_2$ and $T_2$, a complete
representation for $t = 2$. In order to build the data structure for $t = 4$, we continue this
process, by inserting $T_2$ in $B_3$ and retrieving $T_3$ from $T_1$
(stored on disk). It should be noted that to obtain $T_i$
we need $T_{i-2}$, and by always storing it on disk we guarantee
not to use more memory than the size of the final structure. 
The set $T_t$ (that is, $T_1$, 
$T_2$ or $T_4$ in our experiments)
is stored as a sorted array and is searched by a binary search. We
found this implementation more efficient than a hash table. 

Assessing the query time is done through the procedure of graph
traversal, as it is implemented in
\cite{DBLP:conf/wabi/ChikhiR12}. Since the procedure is identical and
independent on the data structure, the time spent on graph traversal
is a faithful estimator of the query time. 

We compare three versions: $t=1$ (i.e. the version of
\cite{DBLP:conf/wabi/ChikhiR12}), $t=2$ and $t=4$. 

\subsection{\emph{E.coli} dataset, varying $k$}
In this set of tests, our main goal was to evaluate the influence of
the $k$-mer size on principal parameters:  the size of
the whole data structure, the size of the set $T_t$, the graph
traversal time, and the time of construction of the data structure. 
We retrieved 10M \emph{E. coli} reads of 100bp
from the \emph{Short Read Archive} (ERX008638) without read pairing
information and 
extracted all $k$-mers occurring at least two times. The total number
of $k$-mer considered varied, depending on the value of $k$, from
6.967.781 ($k=15$) to 5.923.501 ($k = 63$).  We ran each version, 1
Bloom (\cite{DBLP:conf/wabi/ChikhiR12}), 2 Bloom and 4 Bloom, for
values of $k$ ranging from $16$ to $64$. The results are shown in
Fig.~\ref{fig:ecoli}.

The total size of the structures in bits per stored $k$-mer, i.e. the size of
$B_1$ and $T_1$ (respectively, $B_1, B_2$,$T_2$ or $B_1, B_2, B_3,
B_4$,$T_4$)
is shown in
Fig.~\ref{fig:bits_per_kmer}. As expected, the space for 4 Bloom
filters is the smallest for all values of $k$ considered, showing a
considerable improvement, ranging from 32\% to 39\%, over
the version of \cite{DBLP:conf/wabi/ChikhiR12}. Even the version with
just 2 Bloom filters shows
an improvement of at least 20\% over \cite{DBLP:conf/wabi/ChikhiR12}, for all
values of $k$. Regarding the influence of the $k$-mer size on the
structure size, we observe that for 4 Bloom filters the structure size
is almost constant, the minimum value is 8.60 and the largest is 8.89,
an increase of only 3\%. For 1 and 2 Bloom the same pattern is seen: a
plateau from $k = 16$ to $32$, a jump for $k=33$ and another plateau
from $k = 33$ to $64$. The jump at $k=32$ is due to switching from
64-bit to 128-bit representation of $k$-mers in the table $T_t$. 

The traversal times for each version is shown in
Fig.~\ref{fig:traversal}. 
The fastest version is 4 Bloom, showing an improvement over \cite{DBLP:conf/wabi/ChikhiR12} 
of 18\% to 30\%, followed by 2 Bloom. 
This result is surprising and may seem counter-intuitive, as we have
four filters to apply to the queried $k$-mer rather than a single
filter as in \cite{DBLP:conf/wabi/ChikhiR12}. However, the size of
$T_4$ (or even $T_2$) is much smaller than $T_1$, as the size of
$T_i$'s decreases exponentially. As $T_t$ is stored in an array, the time economy in searching $T_4$
(or $T_2$) compared to $T_1$ dominates the time lost on querying
additional Bloom filters, which explains the overall gain in query time. 
%

As far as the construction time is concerned
(Fig.~\ref{fig:construction}), our versions yielded also a faster
construction, with the 4 Bloom version being 5\% to 22\% faster than
that of \cite{DBLP:conf/wabi/ChikhiR12}. The gain is explained by the
time required for sorting the array storing $T_t$, which is much
higher for $T_0$ than for $T_2$ or $T_4$. However, the gain is less significant
here, and, on the other hand, was not observed for bigger datasets (see
Section~\ref{human}).

\begin{figure}[Htbp]
  \center
  \subfigure[]{\includegraphics[width=5.9cm]{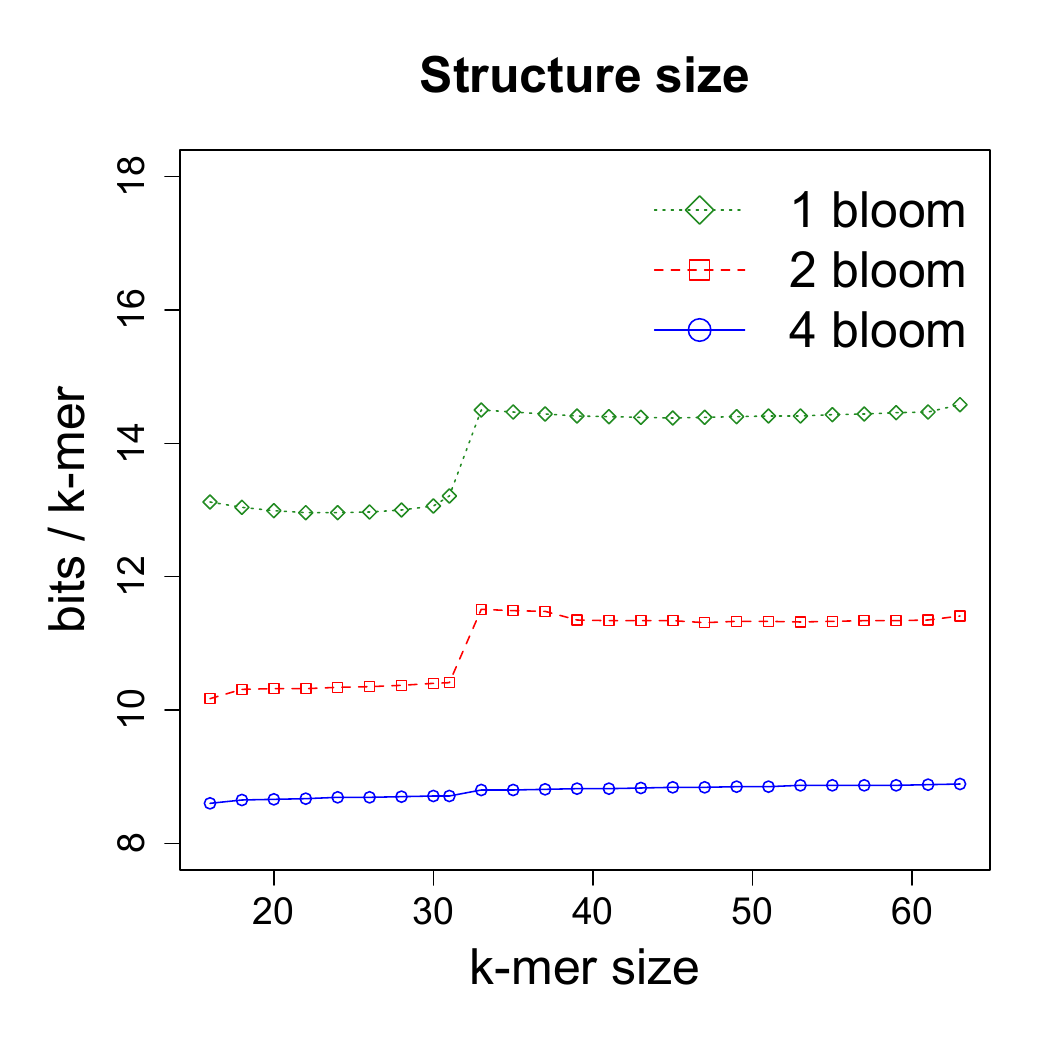}\label{fig:bits_per_kmer}}
  \subfigure[]{\includegraphics[width=5.9cm]{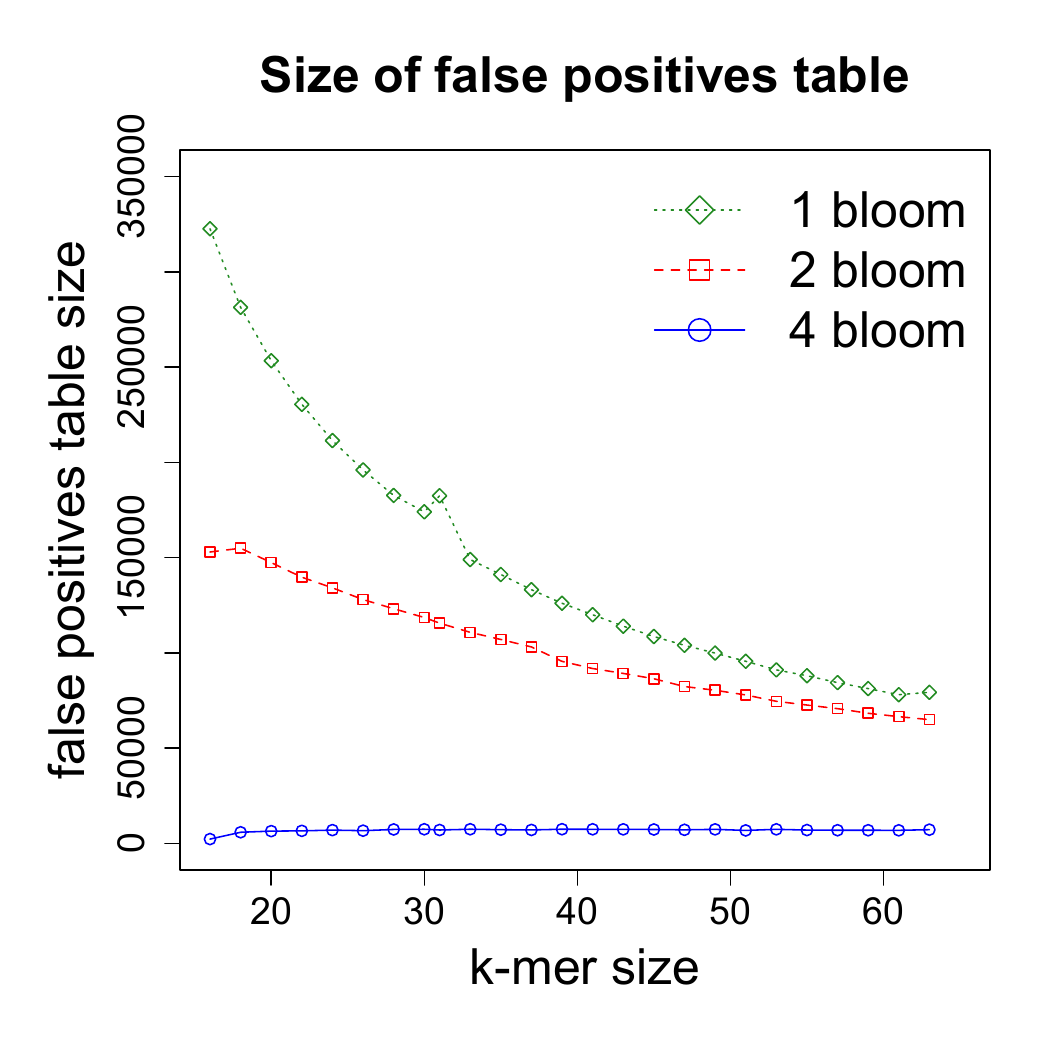}\label{fig:false_positive}}
  \subfigure[]{\includegraphics[width=5.9cm]{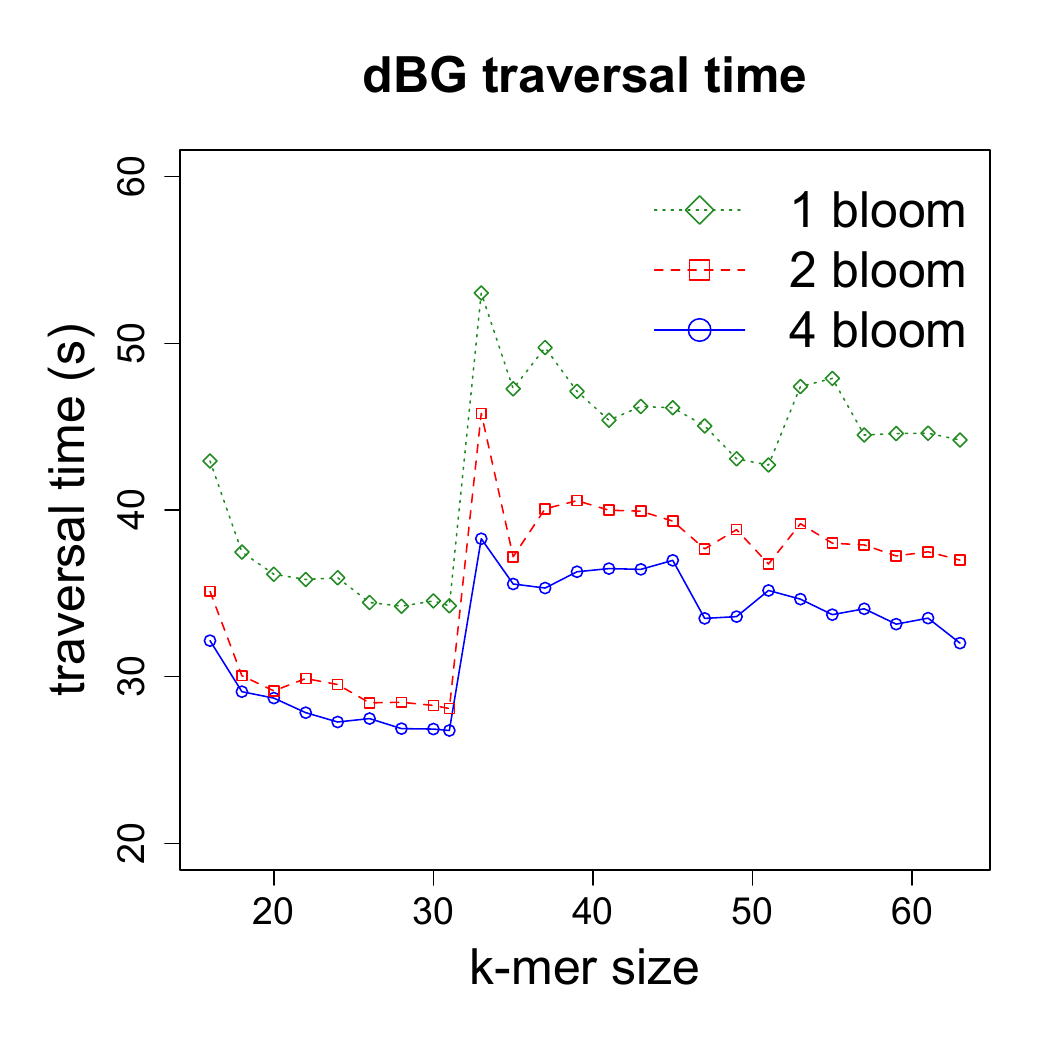} \label{fig:traversal}}
  \subfigure[]{\includegraphics[width=5.9cm]{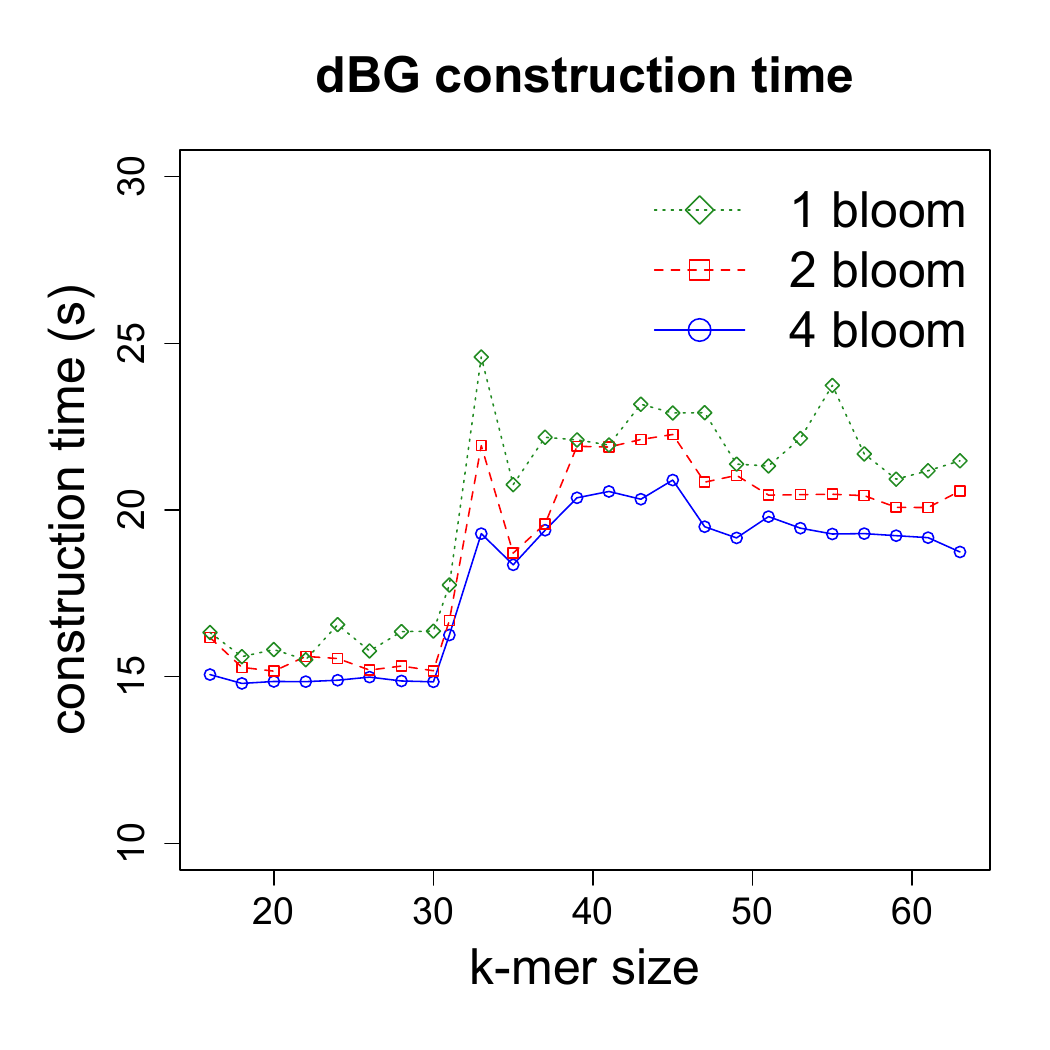}\label{fig:construction}} \\
  \caption{Results for 10M E.coli reads of 100bp using several values
    of $k$. The \emph{1 Bloom} version corresponds to the one
    presented in \cite{DBLP:conf/wabi/ChikhiR12}. (a) Size of the
    structure in bits used per $k$-mer stored. (b)
    Number of false positives stored in $T_1$, $T_2$ or $T_4$ for 1, 2
    or 4 Bloom filters, respectively. (c) De Bruijn graph
    construction time, excluding $k$-mer counting step. (d) De Bruijn
    graph traversal time, including branching $k$-mer indexing. }  \label{fig:ecoli}
\end{figure}

\subsection{\emph{E. coli} dataset, varying coverage} \label{subsec:cov}
From the complete \emph{E. coli} dataset ($\approx$44M reads) from
the previous section, we selected several samples ranging from 5M to 40M
reads in order to assess the impact of the coverage on the size of the
data structures. This strain \emph{E. coli} (K-12 MG1655) is estimated to
have a genome of 4.6M bp~\cite{ecoli}, implying that a sample of 5M
reads (of 100bp) corresponds to $\approx$100X coverage. We
set $d = 3$ and $k = 27$. The results are shown in
Fig.~\ref{fig:ecoli_coverage}. As expected, the 
memory consumption per $k$-mer
remains almost constant for increasing 
coverage, with a slight decrease for 2 and 4 Bloom. The best
results are obtained with the 4 Bloom version, an improvement of 33\%
over the 1 Bloom version of \cite{DBLP:conf/wabi/ChikhiR12}. On the other hand, the number of distinct $k$-mers 
increases markedly (around 10\% for each 5M reads) with increasing coverage, see Fig.~\ref{fig:ecoli_cov:kmers}. 
This is due to
sequencing errors: an increase in coverage implies more errors with
higher coverage, which are not removed by our cutoff
$d = 3$. This suggests that the value of $d$ should be chosen
according to the coverage of the sample. Moreover, in the case where
read qualities are available, a quality control pre-processing step may
help to reduce the number of sequencing errors.

\begin{figure}[Htbp]
  \center 
  \subfigure[]{\includegraphics[width=5.9cm]{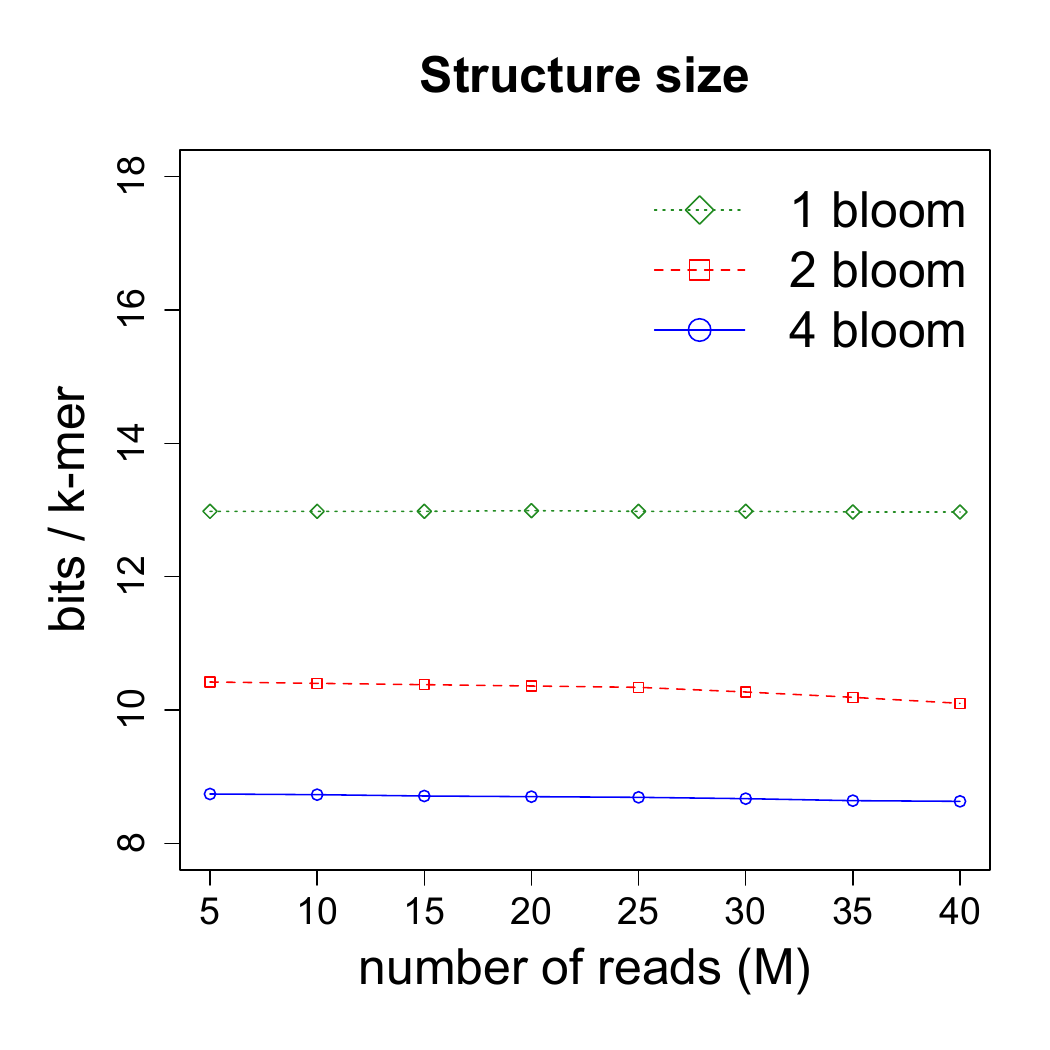}}
  \subfigure[]{\includegraphics[width=5.9cm]{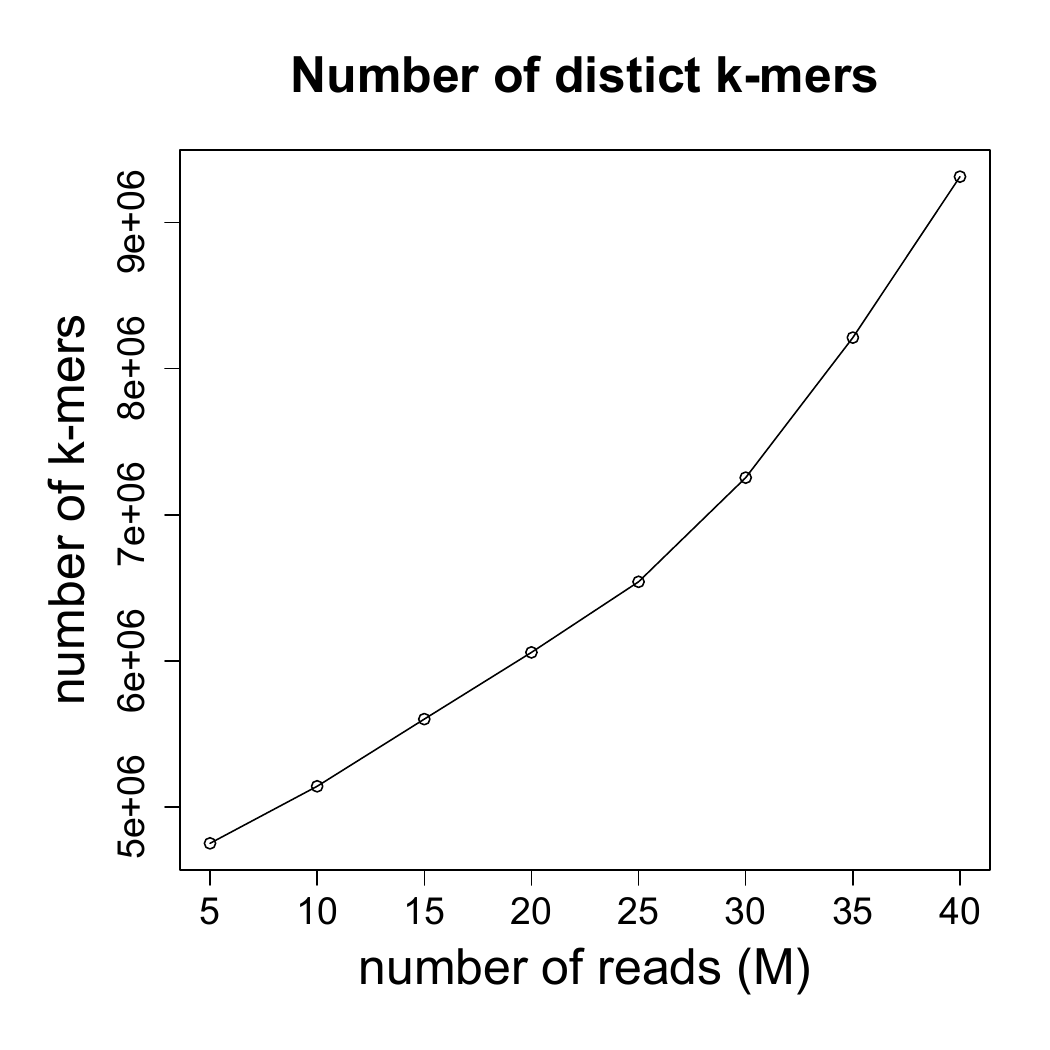} \label{fig:ecoli_cov:kmers}}
  \caption{Results for \emph{E.coli} reads of 100bp using $k =
    27$. The \emph{1 Bloom} version corresponds to the one presented
    in \cite{DBLP:conf/wabi/ChikhiR12}. (a) Size of the structure in
    bits used per $k$-mer stored. (b) Number of distinct $k$-mers.}\label{fig:ecoli_coverage}
\end{figure}

\subsection{Human dataset}
\label{human}
We also compared 2 and 4 Bloom versions with the 1 Bloom version of
\cite{DBLP:conf/wabi/ChikhiR12} on a large dataset. For that, we
retrieved 564M Human reads of 100bp (SRA: SRX016231) without pairing
information and
discarded the reads occurring less than 3 times.  The dataset
corresponds to $\approx$17X coverage. A total of 2.455.753.508
$k$-mers were indexed. We ran each version, 1 Bloom
(\cite{DBLP:conf/wabi/ChikhiR12}), 2 Bloom and 4 Bloom with $k =
23$. The results are shown in Table~\ref{tab:human}.

The results are in general consistent with the previous tests on
\emph{E.coli} datasets. There is an improvement of 34\% (21\%) for the
4 Bloom (2 Bloom) in the size of the structure. The graph traversal is
also 26\% faster in the 4 Bloom version. However, in contrast to the
previous results, the graph construction time {increased} by 10\% and
7\% for 4 and 2 Bloom versions respectively, when compared to the 1
Bloom version. This is due to the fact that disk writing/reading
operations now dominate the time for the graph construction, and 2 and
4 Bloom versions generate more disk accesses than 1 Bloom.  As stated
in Section~\ref{sec:imple}, when constructing the 1 Bloom structure,
the only part written on the disk is $T_1$ and it is read only once to
fill an array in memory. For 4 Bloom, $T_1$ and $T_2$ are written to
the disk, and $T_0$ and $T_1$ are read at least one time each to build
$B_2$ and $B_3$. Moreover, since the size coefficient of $B_1$
reduces, from $r = 11.10$ in 1 Bloom to $r = 5.97$ in 4 Bloom, the
number of false positives in $T_1$ increases.

\begin{table}[htbp]
\begin{center}
\begin{tabular}{|c|c|c|c|}
\hline
Method                          & 1 Bloom & 2 Bloom & 4 Bloom \\
\hline \hline
Construction time (s)           &  40160.7 & 43362.8 & 44300.7 \\
\hline
Traversal time (s)              &  46596.5 & 35909.3 & 34177.2 \\
\hline
$r$ coefficient                      &  11.10 &  7.80    & 5.97 \\
\hline
\multirow{4}{*}{Bloom filters size (MB)}  &  $B_1 = 3250.95$ & $B_1 = 2283.64$ & $B_1 = 1749.04$  \\
                                          &                  & $B_2 = 323.08$  & $B_2 = 591.57$ \\
                                          &                  &                 & $B_3 = 100.56$  \\
                                          &                  &                 & $B_4 = 34.01$ \\
\hline
False positive table size (MB)  &  $T_1 = 545.94$ & $T_2 = 425.74$ & $T_4 = 36.62$ \\
\hline
Total size (MB)                 &  3796.89 & 3032.46 & 2511.8 \\
\hline
\bf Size (bits/$k$-mer)             & \bf 12.96 & \bf 10.35 & {\bf 8.58} \\
\hline
\end{tabular}
\end{center}
\caption{Results of 1, 2 and 4 Bloom filters version for 564M Human
  reads of 100bp using $k = 23$. The \emph{1 Bloom} version
  corresponds to the one presented in
  \cite{DBLP:conf/wabi/ChikhiR12}. }\label{tab:human}
\end{table}

\section{Discussion and Conclusions}
Using cascading Bloom filters for storing de Bruijn graphs brings a
clear advantage over the single-filter method of
\cite{DBLP:conf/wabi/ChikhiR12}. In terms of memory consumption, which
is the main parameter here, we obtained an improvement of around
30\%-40\% in all our experiments. Our data structure takes 8.5 to 9 bits
per stored $k$-mer, compared to 13 to 15 bits by the method of
\cite{DBLP:conf/wabi/ChikhiR12}. 
This confirms our analytical estimations. The above results were obtained using only
four filters and are very close to the estimated optimum (around 8.4
bits/$k$-mer) produced by the infinite number of filters. An
interesting characteristic of our method is that the memory grows insignificantly
with the growth of $k$, even slower than with the method of \cite{DBLP:conf/wabi/ChikhiR12}. 
Somewhat surprisingly, we also obtained a significant
decrease, of order 20\%-30\%, of query time. The construction time of
the data structure varied from being 10\% slower (for the human
dataset) to 22\% faster (for the bacterial dataset). 

As stated previously, another compact encoding of de Bruijn graphs has
been proposed in \cite{DBLP:conf/wabi/BoweOSS12}, however no
implementation of the method was made available. For this reason, we
could not experimentally compare our method with the one of
\cite{DBLP:conf/wabi/BoweOSS12}. We remark, however, that the space
bound of \cite{DBLP:conf/wabi/BoweOSS12} heavily depends on the number
of reads (i.e. coverage), while in our case, the data structure size
is almost invariant with respect to the coverage
(Section~\ref{subsec:cov}).

An interesting prospect for further possible improvements of our
method is offered by work \cite{DBLP:conf/csr/Porat09}, where an
efficient replacement to Bloom filter was introduced. The results of
\cite{DBLP:conf/csr/Porat09} suggest that we could hope to reduce the
memory to about $5$ bits per $k$-mer. However, 
there exist obstacles on this way: an implementation of such a
structure would probably result in a significant construction and query
time increase. This issue remains, however, to be more carefully studied. 


\paragraph{Acknowledgements} Part of this work has been done during
the visit of KS to LIGM in France, supported by the CNRS
French-Russian exchange program in Computer Science. GK has been
partly supported by the ABS2NGS grant of the French gouvernement
(program {\em Investissement d'Avenir}) as well as by a Marie-Curie
Intra-European Fellowship for Carrier Development. GS was supported by
the ERC Advanced Grant Sisyphe held by Marie-France Sagot.


\end{document}